\documentclass[11pt]{amsart}

\usepackage{amsmath}
\usepackage{amssymb}
\usepackage{amsthm}

\usepackage{algorithm}
\usepackage{algorithmicx}
\usepackage{bm}
\usepackage{comment}
\usepackage{cite}
\usepackage{color}
\usepackage{dsfont}
\usepackage{graphicx}

\usepackage[colorlinks, citecolor=red]{hyperref}

\newtheorem{thm}{Theorem}[section]
\newtheorem{lem}[thm]{Lemma}

\newtheorem{cor}[thm]{Corollary}
\newtheorem{defn}[thm]{Definition}

\newtheorem{rem}[thm]{Remark}

\newcommand{\ip}[2]{\langle #1, #2 \rangle} 
\newcommand{\norm}[1]{\lVert #1 \rVert}
\newcommand{\supp}{\operatorname{supp}}
\newcommand{\rank}{\operatorname{rank}}

\newcommand{\e}{\mathrm{e}}
\newcommand{\iu}{\mathrm{i}}
\newcommand{\T}{\mathrm{T}}

\newcommand{\xkh}[1]{\left(#1\right)}
\newcommand{\dkh}[1]{\left\{#1\right\}}

\newcommand{\nj}[1]{\langle #1 \rangle}

\newcommand{\normf}[1]{\|{#1}\|_F}
\newcommand{\normone}[1]{\|{#1}\|_1}
\newcommand{\norms}[1]{\|{#1}\|_2}
\newcommand{\abs}[1]{\left\lvert#1\right\rvert}

\newcommand{\argmin}[1]{\mathop{\rm argmin}\limits_{#1}}

\newcommand{\A}{{\mathcal A}}

\newcommand{\E}{{\mathbb E}}

\newcommand{\PP}{{\mathbb P}}

\newcommand{\R}{{\mathbb R}}
\newcommand{\Rn}{{\mathbb R}^n}
\newcommand{\C}{{\mathbb C}}
\newcommand{\F}{{\mathbb F}}

\newcommand{\x}{{\widehat{\bm{x}}}}

\newcommand{\y}{{\tilde{\bm{y}}}}

\newcommand{\vx}{{\bm x}}
\newcommand{\vw}{{\bm w}}

\newcommand{\vy}{{ \bm{ y}}}

\newcommand{\vu}{{\bm u}}
\newcommand{\vv}{{\bm v}}
\newcommand{\vz}{{\bm z}}
\newcommand{\vb}{{\bm b}}

\newcommand{\va}{{\bm a}}
\newcommand{\vh}{{\bm h}}

\newcommand{\bh}{ \bar{\bm h}}
\newcommand{\bH}{ \bar{\bm H}}

\title{Affine phase retrieval for sparse signals via $\ell_1$ minimization}
\author{Meng huang}
\address{School of Mathematical Sciences, Beihang University, Beijing, 100191, China} \email{menghuang@buaa.edu.cn}
\thanks{Meng Huang is supported by NSFC grant (12201022).}

\author{Shixiang Sun}
\address{LSEC, ICMSEC, Academy of Mathematics and Systems Science, Chinese Academy of Sciences, Beijing 100190, China;\newline
School of Mathematical Sciences, University of Chinese Academy of Sciences, Beijing 100049, China}
\email{sunshixiang@lsec.cc.ac.cn}

\author{Zhiqiang Xu}
\address{LSEC, ICMSEC, Academy of Mathematics and Systems Science, Chinese Academy of Sciences, Beijing 100190, China;\newline
School of Mathematical Sciences, University of Chinese Academy of Sciences, Beijing 100049, China}
\email{xuzq@lsec.cc.ac.cn}
\thanks{Zhiqiang Xu is supported by the National Science Fund for Distinguished Young Scholars (12025108) and NSFC
(12021001).}

\begin{document}

\begin{abstract}
Affine phase retrieval is the problem of recovering signals from the magnitude-only measurements with a priori information.
In this paper, we use the $\ell_1$ minimization to exploit the sparsity of signals for affine phase retrieval,
showing that $O(k\log(\mathrm en/k))$ Gaussian random measurements are sufficient to recover all $k$-sparse signals by solving a natural $\ell_1$ minimization program,
where $n$ is the dimension of signals.
For the case where measurements are corrupted by noises,
the reconstruction error bounds are given for both real-valued and complex-valued signals.
Our results demonstrate that the natural $\ell_1$ minimization program for affine phase retrieval is stable.
\end{abstract}

\maketitle

\section{Introduction}

\subsection{Problem setup}

Affine phase retrieval for sparse signals aims to recover a $k$-sparse signal $\vx_0 \in \C^n$ from the observed data
\[
    y_j = |\ip{\va_j}{\bm x_0} + b_j| + w_j, \quad j=1, \ldots, m,
\]
where $\va_j \in \C^n, \, j=1, \ldots, m$ are given measurement vectors, $\vb := (b_1, \ldots, b_m)^\T \in \C^m$ is the given bias vector, and $\vw := (w_1, \ldots, w_m)^\T \in \R^m$ is the noise vector.
The affine phase retrieval arises in several practical applications,
such as holography \cite{liebling2003local,latychevskaia,barmherzig,guizar} and Fourier phase retrieval \cite{beinert2015, beinert2018, huangK2016, bendory},
where some side information of signals is a priori known before capturing the magnitude-only measurements.

The aim of this paper is to study the following program to recover $\vx_0$ from $\vy := (y_1, \ldots, y_m)^\T \in \R^m$:
\begin{equation} \label{eq:probset}
    \min_{\bm x \in \mathbb C^n} \norm{\bm x}_1 \quad
    \text{s.t.} \ \norm{|\bm A \vx + \bm b| - \bm y}_2 \le \epsilon,
\end{equation}
where $\bm A := [\va_1, \ldots, \va_m] ^\T \in \C^{m\times n}$.

Particularly, we focus on the following questions:
\begin{itemize}
\item[] {\bf Question} 1:
Assume that $\va_j, \, j=1, \ldots, m$, are Gaussian random measurements with
$m = O(k\log(\mathrm en/k))$.
In the absence of noise, i.e., $\vw=0, \, \epsilon=0$,
is the solution to \eqref{eq:probset} $\vx_0$?
\item[]{\bf Question} 2:  In the noisy scenario, is the program \eqref{eq:probset} stable under small perturbation?
\end{itemize}

For the case where $\vx_0\in \C^n$ is non-sparse,
it was shown that $m \ge 4n-1$ generic measurements are sufficient to guarantee the uniqueness of solutions in \cite{apr_2018},
and several efficient algorithms with linear convergence rate was proposed to recover the non-sparse signals $\vx_0$ from $\vy$ under $m=O(n\log n)$ Gaussian random measurements in \cite{hx_2022}.
However, for the case where $\vx_0$ is sparse, to the best of our knowledges, there is no result about it.

\subsection{Related Works}

\subsubsection{Phase retrieval}

The noisy phase retrieval is the problem of recovering a signal $\vx_0 \in \F^n$, $\F\in \{\R,\C\}$ from the magnitude-only measurements
\[
    y'_j = \abs{\nj{\va_j,\vx_0}} + w_j, \quad j=1, \ldots, m,
\]
where $\va_j \in \F^n$  are given measurement vectors and $w_j \in \R$ are noises.
It arises naturally in many areas such as X-ray crystallography \cite{prco_1990, phase_1991, ppc_1993},
coherent diffractive imaging \cite{oss_2013}, and optics \cite{fpr_1978, pra_1982, proi_2015}.
In these settings, optical detectors record only the intensity of a light wave while losing the phase information.
Note that $\abs{\nj{\va_j,\vx_0}}^2=\abs{\nj{\va_j, \mathrm e^{\mathrm i\theta}\vx_0}}^2$ for any $\theta \in \R$.
Therefore the recovery of $\vx_0$ for the classical phase retrieval is up to a global phase.
In the absence of noise, it has been proved that $m\ge 2n-1$ generic measurements suffice to guarantee the uniqueness of solutions for the real case \cite{balan2006signal},
and $m \ge 4n-4$ for the complex case \cite{saveing4d4, conca2015algebraic,wangxu}, respectively.
Moreover, several efficient algorithms have been proposed to reconstruct $\vx_0$ from $\vy':=[y'_1,\ldots,  y'_m]^\T$,
such as alternating minimization \cite{AltMin}, truncated amplitude flow \cite{TAF}, smoothed amplitude flow \cite{2020a}, trust-region \cite{turstregion},
and the Wirtinger flow (WF) variants \cite{WF, TWF,RWF}.

\subsubsection{Sparse phase retrieval}

For several applications, the underlying signal is naturally sparse or admits a sparse representation after some linear transformation.
This leads to the sparse phase retrieval:
\begin{equation} \label{eq:spr}
    \min_{\bm x \in \F^n } \quad \norm{\bm x}_0 \quad
    \text{s.t.} \ \norm{|\bm A \vx | - \bm  \vy' }_2 \le \epsilon,
\end{equation}
where $\bm A:=[\va_1, \ldots, \va_m]^\T$.
In the absence of noise, it has been established that $m = 2k$ generic measurements are necessary and sufficient for uniquely recovering of all  $k$-sparse signals in the real case,
and  $m\ge 4k-2$ are sufficient in the complex case \cite{wangxu2}.
In the noisy scenario, $O(k \log(\mathrm en/k)) $ measurements suffice for stable sparse phase retrieval \cite{eldar2014phase}.
Due to the hardness of $\ell_0$-norm in \eqref{eq:spr}, a computationally tractable approach to recover $\vx_0$ is by solving the following $\ell_1$ minimization:
\begin{equation}  \label{mo:sprl1}
    \min_{\bm x \in \F^n} \quad  \norm{\bm x}_1 \quad
    \text{s.t.} \ \norm{|\bm A \vx | - \vy'}_2 \le \epsilon.
\end{equation}
For the real case, based on the strong restricted isometry property (SRIP) established by Vladislav and Xu \cite{srip_2016},
the authors in \cite{gao2016stable} proved that,
if $\va_1,\ldots,\va_m \sim 1/\sqrt m \cdot \mathcal N (0, I_n)$ are i.i.d. Gaussian random vectors with $m\ge O( k\log(\mathrm en/k))$,  then the solution $\x\in \R^n$ to \eqref{mo:sprl1} satisfies
\[
    \min \dkh{\norm{\x - \vx_0}, \norm{\x + \vx_0}} \lesssim \epsilon+\frac{\sigma_k(\vx_0)_1}{\sqrt k},
\]
where $\sigma_k(\bm x_0)_1 := \min_{|\supp(\bm x)| \le k} \|\bm x - \bm x_0\|_1$.
Lately, this result was extended to the complex case by employing the ``phaselift'' technique in \cite{pr_2021}.
Specifically, the authors in \cite{pr_2021} showed that, for any $k$-sparse signal $\vx_0 \in \C^n$, the solution $\x \in \C^n$ to the program
\[
    \argmin{\bm x \in \mathbb C^n}  \quad \norm{\bm x}_1 \quad
    \text{s.t.} \ \norm{\mathcal{A}(\vx) - \mathcal{A}(\vx_0)}_2 \le \epsilon
\]
satisfies
\[
    \min_{\theta \in [0, 2\pi)}\|\x - \e^{\mathrm i\theta}\bm{x}_0\|_2 \lesssim \frac{\epsilon}{\sqrt{m}\|\bm x_0\|_2},
\]
provided $\va_1, \ldots, \va_m \sim \mathcal N(0,I_n)$ are i.i.d. complex Gaussian random vectors and $m\ge O( k\log(\mathrm en/k))$.
Here, $\mathcal{A}(\vx) := (|\va_1^* \vx|^2, \ldots, |\va_m^* \vx|^2)$.

\subsubsection{Affine phase retrieval}

The affine phase retrieval aims to recover a signal $\vx_0 \in\F^n$ from the measurements
\[
    y_j =\left|\ip{\va_j}{\bm x_0} + b_j\right|, \quad j=1, \ldots, m,
\]
where $\va_j \in \F^n, \, j=1, \ldots, m$ are measurement vectors, $\vb := (b_1, \ldots, b_m)^\T \in \F^m$ is the bias vector.
The problem can be regarded as the classic phase retrieval with a priori information, and is raised in many areas,
such as holographic phase retrieval \cite{microscopy_1948, microscopy_1949, liebling2003local} and Fourier phase retrieval \cite{beinert2015, beinert2018, huangK2016, bendory}.
In such scenarios, one needs to employ some additional information about the desired signals to ensure the uniqueness of solutions.
Specifically, in holographic optics, a reference signal $\bm r \in \C^k$, whose structure is a priori known, is included in the diffraction patterns alongside the signal of interest $\vx_0\in \C^n$ \cite{latychevskaia, barmherzig, guizar}.
Set $\bm x_0' = (\bm x_0^\T, \bm r^\T)^\T \in \mathbb C^{n+k}$.
Then the magnitude-only measurements we obtain that
\begin{equation*}
    y_j = |\ip{\va'_j}{\bm x_0'}| = |\ip{\va_j}{\bm x_0} + \ip{\va_j''}{\bm r}| = |\ip{\va_j}{\bm x_0} + b_j|, \quad j=1,\ldots,m,
\end{equation*}
where ${\va}'_j = (\va_j^\T, {\va_j''}^\T)^\T \in \C^{n+k}$ are given measurement vectors and $b_j = \ip{\va_j''}{\bm r} \in \C$ are known.
Therefore, the holographic phase retrieval  can be viewed as the affine phase retrieval.

Another application of affine phase retrieval arises in Fourier phase retrieval problem.
For one-dimensional Fourier phase retrieval problem, it usually does not possess the uniqueness of solutions \cite{pro_1963}.
Actually, for a given signal with dimension $n$, beside the trivial ambiguities caused by shift, conjugate reflection and rotation, there still could be $2^{n-2}$ nontrivial solutions.
To enforce the uniqueness of solutions, one approach is to use additionally known values of some entries \cite{beinert2018}, which can be recast as affine phase retrieval.
More related works on the uniqueness of solutions for Fourier phase retrieval can be seen in \cite{edidin2019, fpr_1985}.

\subsection{Our contributions}

In this paper, we focus on the recovery of sparse signals from the magnitude of affine measurements.
Specifically, we aim to recover a $k$-sparse signal $\vx_0\in \mathbb F^n$ ($\mathbb F=\R$ or $\mathbb F=\C$) from the data
\[
    \vy=|\bm{Ax}_0 + \bm b|+\vw,
\]
where $\bm A := [\va_1,\ldots,\va_m]^* \in \mathbb F^{m\times n}$ is the measurement matrix, $\vb \in \F^m$ is the bias vector, and $\vw \in \R^m$ is the noise vector.
Our aim is to present the performance of the following $\ell_1$ minimization program:
\begin{equation} \label{eq:probset22}
    \argmin{\bm x \in \mathbb F^n} \norm{\bm x}_1 \quad
    \text{s.t.} \ \norm{|\bm A \vx + \bm b| - \bm y}_2 \le \epsilon.
\end{equation}
We first introduce the following results:
\begin{thm}\label{th:impos}
Assume that there exists a matrix $\bm A \in \F^{m\times n}$, a vector $\vb \in \F^m$,  a decoder $\Delta : \F^m \rightarrow \F^n$ and positive integers $k_0, \, p, \, q$ such that
\begin{equation}\label{eq:best11}
    \|\Delta(\abs{\bm A\vx+\vb})-\vx\|_p \leq C\cdot \sigma_{k_0}{(\vx)}_q
\end{equation}
holds for all $\vx \in \F^n$ where $C := C_{k_0, p, q}$ is a constant depending on $k_0$, $p$ and $q$. Then $\vb \notin \{\bm A\vz : \vz\in \F^n\}$.
Here, $\sigma_k(\bm x)_q := \min_{|\supp(\bm z)| \le k} \|\bm z - \bm x\|_q$.
\end{thm}
\begin{proof}
We assume that $\vb = \bm A \vz_0$ where $\vz_0 \in \F^n$. We next show that there exits $\vx \in \F^n$ such that \eqref{eq:best11} does not hold.
For the aim of contradiction, we assume that (\ref{eq:best11}) holds.
Since $\sigma_{k_0}{(-\vx)}_q=\sigma_{k_0}{(\vx)}_q$, we have
\begin{equation}\label{eq:best21}
    \|\Delta(\abs{\bm A\vx-\vb})+\vx\|_p=\|\Delta(\abs{\bm A(-\vx)+\vb})-(-\vx)\|_p\leq C \sigma_{k_0}{(\vx)}_q.
\end{equation}
Assume that $\vx_0\in \F^n$ is $k_0$-sparse, i.e. $\sigma_{k_0}(\vx_0)_q=0$. According to \eqref{eq:best11} and \eqref{eq:best21}, we obtain that
\begin{equation}\label{eq:bestmid}
    \Delta(\abs{\bm A\vx_0+\vb})\,\,=\,\, \vx_0,\quad  \Delta(\abs{\bm A\vx_0-\vb})\,\,=\,\, -\vx_0.
\end{equation}
Taking $\vx=r\vx_0+2\vz_0$ in (\ref{eq:best21}), we have
\begin{equation}\label{eq:d1}
    \|\Delta(\abs{\bm A(r\vx_0+2\vz_0)-\vb})+r\vx_0+2\vz_0\|_p\leq C \sigma_{k_0}{(r\vx_0+2\vz_0)}_q\leq C\sigma_{k_0}{(2\vz_0)}_q,
\end{equation}
where  $r>0$.
Observe that
\begin{equation}\label{eq:31}
    \Delta(\abs{\bm A(r\vx_0+2\vz_0)-\vb})=\Delta(\abs{\bm A(r\vx_0)+\vb})=r\vx_0.
\end{equation}
Here, we use $\vx_0$ is $k_0$-sparse. Substituting \eqref{eq:d1} into \eqref{eq:31}, we obtain that
\begin{equation}\label{eq:last1}
    \|2r\vx_0+2\vz_0\|_p\,\,\leq\,\, C\sigma_{k_0}(2\vz_0)_q
\end{equation}
holds for any $r>0$.
Note $\lim_{r\rightarrow \infty}\|2r\vx_0+2\vz_0\|_p=\infty$.
Hence, \eqref{eq:last1} does not hold provided $r$ is large enough. A contradiction!
\end{proof}

For the case where $m\leq n$ and $\bm A$ is full rank, we have $\vb\in \{\bm A\vz : \vz\in \F^n\}$.
According to Theorem \ref{th:impos}, one can not build the instance-optimality result (\ref{eq:best11}) under this setting.

\subsubsection{Real Case}

Our first result gives an upper bound for the reconstruct error of \eqref{eq:probset22} in the real case,
under the assumption of $\va_1, \ldots,\va_m \in \R^n$  being real Gaussian random vectors and $m\ge O(k\log(\mathrm en/k))$.
It means the $\ell_1$-minimization program is stable under small perturbation, even for the approximately $k$-sparse signals.
The result also implies that one can obtain the instance-optimality result if we add some conditions for the signal $\vx$.

\begin{thm} \label{MainThm_rn}
Let $\bm A \in \mathbb{R}^{m \times n}$ be a Gaussian random matrix with entries $a_{jk} \sim \mathcal N(0,1/ m)$.
Let $\bm b$ be a vector satisfying $\alpha < \norm{\bm b_I}_2 < \beta $ for all $I \subseteq [m]$
with $|I| \ge m/2$ where $\alpha,\, \beta>0$ are positive constants.
Assume that $m \ge C a(k+1)\log(\mathrm en/k)$ with $a(k+1) \le n$
where $a > \theta_u/\theta_l$ is a constant and $ \theta_l, \theta_u, C>0$ are positive constants depending only on $\alpha$ and $\beta$.
Then, with probability at least $1 - 4\exp(-c m)$, the following holds:
for any vector $\vx_0 \in \Rn$, the solution $\x$ to
\eqref{eq:probset22}
with $\bm y = |\bm A \bm x_0 + \bm b| + \bm w$ and $\norm{\bm w}_2 \le \epsilon$ obeys
\begin{equation*}
    \norm{\x - \bm x_0}_2 \le K_1\epsilon+ K_2 \frac{\sigma_k(\vx_0)_1}{\sqrt {a(k+1)}},
\end{equation*}
provided $K_1\epsilon + K_2 \frac{\sigma_k(\vx_0)_1}{\sqrt {a(k+1)}} <2$.
Here,
\begin{equation*}
    K_1 := \frac{2\xkh{1+1/\sqrt a}}{\sqrt{\theta_l}-\sqrt{\theta_u}/\sqrt a} >0, \quad K_2 := \sqrt{\theta_u}K_1 + 2,
\end{equation*}
and $\sigma_k(\bm x_0)_1 := \min_{|\supp(\bm x)| \le k} \|\bm x - \bm x_0\|_1$.
\end{thm}

In the absence of noise, i.e., $\vw=0, \, \epsilon=0$, Theorem \ref{MainThm_rn} shows that if $\va_1, \ldots, \va_m \sim 1/\sqrt m \cdot \mathcal N(0, I_n)$ are real Gaussian random vectors and $m\ge O(k\log(\mathrm en/k))$,
then all the $k$-sparse signals $\vx_0 \in \Rn$ could be reconstructed exactly by solving the program \eqref{eq:probset22} under some mild conditions on $\vb$.
We state it as the following corollary:

\begin{cor} \label{MainThm_r}
Let $\bm A \in \mathbb{R}^{m \times n}$ be a Gaussian random matrix with entries $a_{jk}\sim \mathcal N(0,1/m)$,  and $\bm b \in \mathbb{R}^m$ be a vector satisfying
$\alpha  \le  \norm{\bm b_I}_2 \le \beta $ for all $I \subseteq [m]$ with $ |I|\ge m/2 $, where $\alpha$ and $\beta$ are two positive universal constants.
If $m \ge C k\log(\mathrm en/k)$, then with probability at least $1 - 4\exp(-c m)$ it holds: for any $k$-sparse signal $\bm x_0 \in \mathbb{R}^n$, the $\ell_1$ minimization
\begin{equation*}
    \argmin{\bm x \in \mathbb{R}^n} \|\bm x\|_1 \quad  \mathrm{s.t.} \quad  |\bm A \vx + \bm b| = \bm y
\end{equation*}
with $\bm y = |\bm A \vx_0 + \bm b|$ has a unique solution $\bm x_0$. Here $C,\, c>0$ are constants depending only on $\alpha$ and $\beta$.
\end{cor}

\subsubsection{Complex case}

We next turn to consider the estimation performance of \eqref{eq:probset22} for the complex-valued signals.
Let $\mathbb{H}^{n \times n}$ be the set of Hermitian matrix in $\C^{n\times n}$ and $\|\bm H\|_{0, 2}$ denotes the number of non-zero rows in $\bm H$.
Given $\va_1, \ldots, \va_m \in \C^n$ and $b_1,\ldots,b_m \in \C$, we define a linear map $\mathcal A': \bm H' \in \mathbb H^{(n+1)\times (n+1)}\to \R^m$ as follows:
\begin{equation} \label{eq:limAp0}
\mathcal A'(\bm H') = (\va_1'^* \bm H' \va_1', \dots, \va_m'^* \bm H' \va_m'),
\end{equation}
where $\va_j' := \Bigl( \begin{array}{l} \va_j \\  b_j \end{array} \Bigr) \in \C^{n+1}$.

\begin{defn}
We say the linear map $\mathcal A'$ defined in \eqref{eq:limAp0} satisfies the restricted isometry property of order $(r,k)$ with constants $c,\, C >0$ if the following holds
\begin{equation}\label{eq:RIPrk}
    c \norm{\bm H'}_F \le \frac{1}{m}\norm{\mathcal A'(\bm H')}_1 \le C \norm{\bm H'}_F
\end{equation}
for all $\bm H':= \begin{bmatrix} \bm H & \vh \\ \vh^* & 0 \end{bmatrix} \in \mathbb H^{(n+1)\times (n+1)}$
with $\rank(\bm H) \le r$, $\norm{\bm H}_{0, 2} \le k$ and $\norm{\vh}_0 \le k$.
\end{defn}

The following theorem shows that the linear map $\A'$ satisfies the restricted isometry property over low-rank and sparse matrices,
provided $\va_1, \ldots, \va_m \in \C^n$ are i.i.d. complex Gaussian random vectors and $\vb := (b_1,\ldots,b_m)^\T \in \C^m$ satisfies some mild conditions.

\begin{thm} \label{lem_ripc20}
Suppose $\va_1, \ldots, \va_m \sim 1/\sqrt{2}\cdot\mathcal N(0,I_n) + \mathrm i/\sqrt{2}\cdot \mathcal N(0,I_n)$ are i.i.d. complex Gaussian random vectors
and $\vb\in \C^m$ is a independent sub-gaussian random vector (it also may be deterministic) with sub-gaussian norm $\norm{\vb}_{\psi_2} \le C$ and $\E \norm{\vb}_1 \ge c_1 m$, $\E \norm{\vb}_2 \le c_2 \sqrt m $,
where $C,\, c_1,\, c_2 > 0$ are universal constants.
If $m\ge C' k\log(\mathrm en/k)$, then with probability at least $1-5\exp(-c' m)$, the linear map $\A'$ defined in \eqref{eq:limAp0} obeys
\begin{equation*}
    \frac{\theta^-} {12} \norm{\bm H'}_F \le \frac{1}{m}\norm{\mathcal A'(\bm H')}_1 \le 3\theta^+ \norm{\bm H'}_F
\end{equation*}
for all $\bm H':= \begin{bmatrix} \bm H & \vh \\ \vh^* & 0 \end{bmatrix} \in \mathbb H^{(n+1)\times (n+1)}$ with $\rank(\bm H) \le 2$, $\norm{\bm H}_{0, 2} \le k$ and $\norm{\vh}_0 \le k$.
Here, $\theta^-:=\min(1, c_1/\sqrt 2)$, $\theta^+:=\max(\sqrt 6, c_2)$, and $C',\, c'>0$ are constants depending only on $c_1,\, c_2$.
\end{thm}

With abuse of notation, we denote $\mathcal A'(\vx'):=\mathcal A'(\vx'\vx'^*)$ for any vector $\vx'\in \C^{n+1}$. Then we have

\begin{thm} \label{th:stubon0}
Assume that the linear map $\mathcal A'(\cdot)$ satisfies the RIP condition (\ref{eq:RIPrk}) of order $(2, 2ak)$ with constants $c,\, C>0$.  For any $k$-sparse signal $\vx_0 \in \C^n$, if
\[
c - C\xkh{ \frac{4}{\sqrt a}+\frac 1a} >0,
\]
then the solution $\x\in \C^n$ to
\begin{equation*}
    \argmin{\vx \in \C^n} \quad \norm{\vx}_1 \quad \mbox{\rm s.t.} \quad \norm{\mathcal A'(\vx')-\y}\le \epsilon \quad \mbox{and} \quad \vx'=(\vx^\T, 1)^\T
\end{equation*}
with $\y=\mathcal A'(\vx'_0)+\vw$, $\norm{\vw} \le \epsilon$ and $\vx'_0=(\vx_0^\T, 1)^\T$ obeys
\[
    \min_{\theta\in \R} \xkh{ \norms{\x-\mathrm{e}^{\mathrm i\theta}\vx_0} + |1-\mathrm{e}^{\mathrm i\theta} |} \le \frac{C_0 \epsilon}{\xkh{\norm{\vx_0}+1}\sqrt m},
\]
where
\[
    C_0:=2\sqrt 2 \cdot \frac {\frac1a +\frac{4}{\sqrt a}+1}{c - C\xkh{ \frac{4}{\sqrt a}+\frac 1a}}.
\]
\end{thm}

Based on Theorem \ref{lem_ripc20}, if $\va_1,\ldots,\va_m \in \C^n$ are i.i.d. complex Gaussian random vectors and $m\ge C' ak\log(\mathrm en/ak)$,
then with high probability the linear map $\mathcal A'$ defined in \eqref{eq:limAp0} satisfies RIP conditions of order $(2,2ak)$ with constants $c=\theta^-/12$ and $C=3\theta^+$ under some mild conditions on $\vb$.
For the noiseless case where $\vw=0,\,\epsilon=0$, taking the constant $a> (8C/c)^2$ and combining with Theorem \ref{th:stubon0}, we can obtain the following result.
\begin{cor}
Suppose $\va_1, \ldots, \va_m \sim 1/\sqrt{2}\cdot \mathcal N(0,I_n)+\mathrm{i}/\sqrt{2}\cdot\mathcal N(0,I_n)$ are i.i.d. complex Gaussian random vectors and $\vb\in \C^m$ is a independent sub-gaussian random vector (it also may be deterministic)
with sub-gaussian norm $\norm{\vb}_{\psi_2} \le C$ and  $ \E \norm{\vb}_1 \ge c_1 m$, $\E \norm{\vb}_2 \le c_2 \sqrt m$, where $C,\, c_1,\, c_2 > 0$ are universal constants.
If $m \ge C'' k\log(\mathrm en/k)$, then with probability at least $1-5\exp(-c'' m)$, then the solution to
\[
    \argmin{\vx \in \C^n} \quad \norm{\vx}_1 \quad \mbox{s.t.} \quad |\bm A\vx+\vb|=|\bm A\vx_0+\vb|
\]
is $\vx_0$ exactly. Here, $C'',\, c'' > 0$ are constants depending only on $c_1,\,c_2$.
\end{cor}

\begin{rem}
We give an upper bound for $\min_{\theta\in\R} \xkh{\norms{\x-\e^{\mathrm{i}\theta}\vx_0} + |1-\e^{\mathrm i\theta}|}$ in Theorem \ref{th:stubon0}.
However, since the affine phase retrieval can recover a signal exactly (not just up to a global phase), therefore, one may wonder: is there a stable recovery bound for $\norms{\x-\vx_0}$? 
We truly believe that the answer is no, especially for the case where the noise vector $\norms{\vw} \gtrsim \sqrt m$. We defer the proof of it for the future work.
\end{rem}

\subsection{Notations}
Throughout the paper, we denote $\vx\sim \mathcal N(0,I_n)$ if $\vx\in \R^n$ is a standard Gaussian random vector.
A  vector $\bm x$ is $k$-sparse if there are at most $k$ nonzero entries of $\vx$. For simplicity, we denote $[m] := \{1, \dots, m\}$.
For any subset $I \subseteq [m]$, let  $\bm A_I = \begin{bmatrix} \va_j: j \in I \end{bmatrix}^*$ be the submatrix whose rows are generated by $\bm A = \begin{bmatrix} \va_1,\ldots, \va_m \end{bmatrix}^*$.
Denote $\sigma_k(\bm x_0)_p := \min_{|\supp(\bm x)| \le k} \|\bm x - \bm x_0\|_p$ as the best $k$-term approximation error of $\bm x_0$ with respect to $\ell_p$ norm.
For a complex number $b$,  we use $b_{\Re}$ and $b_{\Im}$ to denote the real and imaginary part of $b$, respectively.  For any $A,B\in \R$, we use $ A \lesssim B$
to denote $A\le C_0 B$ where $C_0\in \R_+$ is an  absolute constant.  The notion
$\gtrsim$ can be defined similarly.  Throughout  this paper, $c$, $C$ and the subscript (superscript) forms of them denote constants whose values vary with the context.


\section{Proof of Theorem \ref{MainThm_rn}}

In this section, we consider the estimation performance of the $\ell_1$-minimization program \eqref{eq:probset22} for the real-valued signals.
To begin with, we need the following definition of strong RIP condition, which was introduced by Vladislav and Xu \cite{srip_2016}.
\begin{defn}[Strong RIP in \cite{srip_2016}] \label{eq:strongrip}
The matrix $\bm A\in \R^{m\times n}$ satisfies the Strong Restricted Isometry Property (SRIP) of order $k$ and constants $\theta_l,\, \theta_u >0 $ if the following inequality
\begin{equation*}
    \theta_{l}\norm{\vx}^2 \le \min_{I\subset [m],\abs{I}\ge m/2} \norm{\bm A_I \vx}^2\le \max_{I\subset [m],\abs{I}\ge m/2} \norm{\bm A_I \vx}^2 \le \theta_{u}\norm{\vx}^2
\end{equation*}
holds for all $k$-sparse signals $\vx\in \R^n$. Here, $\bm A_I$ denotes the sub-matrix of $\bm A$ whose rows with indices in $I$ are kept, $[m]:=\{1,\ldots,m\}$ and $\abs{I}$ denotes the cardinality of $I$.
\end{defn}

It has been shown that if $\bm A \in \R^{m\times n}$ is a real Gaussian random matrix with entries $a_{k,j}\sim \mathcal N(0,1/m)$, then $\bm A$ satisfies the strong RIP with high probability, as stated below.

\begin{lem} \textnormal{(Theorem 2.1 in \cite{srip_2016})} \label{lem_srip}
Suppose that $t>1$ and that $\bm A\in \R^{m\times n}$ is a Gaussian random matrix with entries $a_{k,j}\sim \mathcal N(0,1/m)$.  Let $m=O(tk \log(\e n/k))$ where $k \in [1,d]\cap {\mathbb Z}$ and $t\geq 1$ is a constant.
Then there exist constants $\theta_l, \theta_u$ with $0<\theta_l< \theta_u<2$, independent with $t$,
such that $\bm A$ satisfies SRIP of order $t\cdot k$ and constants $\theta_l,\, \theta_u$ with probability at least $1-\exp(-cm)$, where $c>0$ is a universal constant.
\end{lem}

The following lemma plays a key role in the proof of Theorem \ref{MainThm_rn},
which indicates the matrix $\begin{bmatrix} \bm A & \bm b \end{bmatrix} \in \mathbb R^{m \times (n+1)}$ satisfies strong RIP with high probability
under some mild conditions on $\bm A\in \R^{m\times n}$ and $\vb\in \R^m$.

\begin{lem} \label{lem_srip2}
Let $\bm A \in \mathbb R^{m \times n}$ be a Gaussian random matrix with entries $a_{k, j}\sim \mathcal N(0,1/m)$.  Suppose that the vector $\bm b \in \mathbb R^m$ satisfies $\alpha < \norm{\bm b_I}_2 < \beta $ for all $I \subseteq [m]$
with $|I| \ge m/2$, where $\alpha$ and $\beta$ are two positive constants.  Set $\bm A' := \begin{bmatrix} \bm A & \bm b \end{bmatrix} \in \mathbb R^{m \times (n+1)}$.
If $m \ge Ct(k+1)\log(\mathrm en/k)$ with $t(k+1) \le n$ and $1 < t \in \mathbb Z$,
then the matrix $\bm A'$ satisfies the strong RIP of order $tk+1$ and constants $\theta_l',\, \theta_u'$
with probability at least $1 - 4\exp(-c'm)$, where  $C,\, c'>0$ are constants depending only on $\alpha$ and $\beta$.
Here, $\theta'_l =  0.99\min\{\theta_l, \alpha^2\}$ and $\theta'_u = 1.01 \max\{\theta_u, \beta^2\} $ with $\theta_l, \theta_u$ being defined in Lemma \ref{lem_srip}.
\end{lem}

\begin{proof}
From the definition, it suffices to show there exist constants $\theta'_l, \theta'_u >0 $ such that the following inequality
\begin{equation}\label{eq:strongrip22}
\theta'_{l}\norm{\vx'}^2 \le \min_{I\subset [m],\abs{I}\ge m/2} \norm{\bm A'_I \vx'}^2\le \max_{I\subset [m],\abs{I}\ge m/2} \norm{\bm A'_I \vx'}^2 \le \theta'_{u}\norm{\vx'}^2
\end{equation}
holds for all $(tk+1)$-sparse signals $\vx' \in \R^{n+1}$. To this end, we denote $\bm x' = (\bm{x}^\T, z)^\T$, where $\bm x \in \mathbb R^n$ and  $z \in \mathbb R$.
We first consider the case where $z=0$.
From Lemma \ref{lem_srip}, we know that if $m \gtrsim t(k+1)\log(\mathrm en/(k+1))$ and $t>1$, then
there exist two positive constants $\theta_l, \theta_u \in (0,2)$ such that
\begin{equation} \label{eq_srip2_neq}
    \theta_l \|\bm x\|_2^2 \le \min_{I \subseteq [m], |I |\ge m/2} \|\bm A_I \bm x\|_2^2
    \le \max_{I \subseteq [m], |I|\ge m/2}\|\bm A_I \bm x\|_2^2 \le \theta_u \|\bm x\|_2^2
\end{equation}
holds for all $(tk+1)$-sparse vector $\bm x \in \mathbb R^n$ with probability at least $1 - \exp(-cm)$.
Here, $c>0$ is a universal constant.
Note that $\bm A' \bm x' = \bm A \bm x$. We immediately obtain \eqref{eq:strongrip22} for the case where $z=0$.

Next, we turn to the case where $z \ne 0$. A simple calculation shows that
\begin{equation} \label{eq_srip2_bi}
    \|\bm A'_I \bm x'\|_2^2 = \|\bm A_I \bm x + z\bm{b}_I\|_2^2 = \|\bm A_I \bm x\|_2^2 + 2z\ip{\bm A_I \bm x}{\bm b_I} + z^2\|\bm b_I\|_2^2
\end{equation}
for any $I \subseteq [m]$. Denote $\bm A = \begin{bmatrix} \va_1,\ldots, \va_m \end{bmatrix}^\T$. Note that $\sqrt m \va_j \sim \mathcal N(0, I_n)$.
Taking $\zeta = \frac{\min(\theta_l,\alpha^2)}{200\beta}$ in Lemma \ref{lem_bjajx},
we obtain that there exists a constant $C>0$ depending only on $\theta_l,\, \alpha,\, \beta$ such that when $m\ge Ct(k+1)\log(\e n/k)$,
with probability at least $1-3\exp(-c_1 m)$, it holds
\begin{equation} \label{eq_srip2_cbs}
    |\ip{\bm A_I \bm x}{\bm b_I} | =  |\ip{\bm A \vx}{\bm b_I}| \le \frac{\min\{\theta_l,\alpha^2\}}{200\beta}  \|\bm x\|_2\|\bm b\|_2
\end{equation}
for all $(tk+1)$-sparse vectors $\bm x$ and all $I \subseteq [m]$.
Here, we view ${\bm b}_I={\bm b}\mathbb{I}_I\in \R^m$ (${\mathbb I}_I(j)=1$ if $j\in { I}$ and $0$ if $j\notin { I}$), and $c_1 > 0$ is a constant depending only on $\theta_l,\, \alpha,\, \beta$.
Note that the vector $\bm b$ satisfies
\begin{equation}  \label{eq_srip2_b}
    \alpha \le \norm{\bm b_I}_2 \le \beta
\end{equation}
for all $I \subseteq [m]$ with $|I|\ge m/2$.
Putting \eqref{eq_srip2_neq}, \eqref{eq_srip2_cbs} and \eqref{eq_srip2_b} into \eqref{eq_srip2_bi},
we obtain that when $m \ge Ct(k+1)\log(\e n/k)$, with probability at least $1-4\exp(-cm)$, the following two inequalities
\begin{equation*}
    \|\bm A'_I \bm x'\|_2^2
    \ge \theta_l \|\bm x\|_2^2 - 2|z|\frac{\min\{\theta_l, \alpha^2\}}{200\beta}\|\bm x\|_2 \beta + \alpha^2 z^2
    \ge 0.99\min\{\theta_l, \alpha^2\} \|\bm x'\|_2^2,
\end{equation*}
and
\begin{eqnarray*}
    \|\bm A'_I \bm x'\|_2^2  \le  \theta_u \|\bm x\|_2^2 + 2|z|\frac{\min\{\theta_l, \alpha^2\}} {200\beta}\|\bm x\|_2 \beta + \beta^2 z^2 \le  1.01  \max\{\theta_u, \beta^2\}  \|\bm x'\|_2^2
\end{eqnarray*}
hold for all $(tk+1)$-sparse vector $\bm x' \in \mathbb R^{n+1}$ and for all $I \subseteq [m]$ with $|I|\ge m/2$.
Here, $c > 0$ is a constant depending only on $\theta_l,\, \alpha,\, \beta$. In other words, we have
\begin{equation*}
    \theta'_l \norm{\bm x'}_2^2 \le \min_{I \subseteq [m], |I |\ge m/2} \|\bm A'_I \bm x'\|_2^2
    \le \max_{I \subseteq [m], |I |\ge m/2} \|\bm A'_I \bm x' \|_2^2 \le \theta'_u \|\bm x'\|_2^2
\end{equation*}
for all $(tk+1)$-sparse vector $\bm x'$ with probability at least $1 - 4\exp(-cm)$.
Here, $\theta'_l =  0.99\min\{\theta_l, \alpha^2\}$ and $\theta'_u = 1.01  \max\{\theta_u, \beta^2\}$.
Combining the above two cases and noting that $\theta_l$, $\theta_u >0$ are universal constants, we complete the proof.
\end{proof}

Based on Lemma \ref{lem_srip2}, we are now ready to present the proof of Theorem \ref{MainThm_rn}.

\begin{proof}[Proof of Theorem \ref{MainThm_rn}]
Denote $\bm A' = \begin{bmatrix} \bm A & \bm b \end{bmatrix}$,
$\x' = (\x^\T, 1)^\T$ and $\vx'_0 = (\bm x_0^\T, 1)^\T$. 
Set
\[
    I := \{j: (\ip{\va_j}{\x} + b_j)(\ip{\va_j}{\bm x_0} + b_j) \ge 0\}.
\]
We next divide the proof into the following two cases.

{\bf Case 1: $|I| \ge m/2$}.
Set $\vh = \x' - \vx'_0$. For any $a>1$, we decompose $\vh$ into the sum of $\vh_{T_0}, \vh_{T_1}, \dots$,
where $T_0$ is an index set
which consists the indices of the $k+1$ largest coordinates of $\vx'_0$ in magnitude,
$T_1$ is the index set corresponding to the $a(k+1)$ largest coordinates of $\vh_{T_0^c}$ in magnitude,
$T_2$ is the index set corresponding to the $a(k+1)$ largest coordinates of $\vh_{(T_0 \cup T_1)^c}$ in magnitude, and so on.
For simplicity, we denote $T_{jl} := T_j \cup T_l$.  To prove the theorem,  we only need to give an upper bound for $\norms{\vh}$. Observe that
\begin{equation} \label{cla:hR0}
    \norms{\vh} \le \norms{\vh_{T_{01}}} + \norms{\vh-\vh_{T_{01}}}.
\end{equation}

We claim that, when $m\ge C (a+1)(k+1)\log(\e n/k)$, with probability at least $1-4\exp(-c m)$, it holds that
\begin{equation}\label{cla:hR1}
    \norms{\vh-\vh_{T_{01}}} \le \frac1{\sqrt a}\norm{\vh_{T_{01}}}_2 + \frac{2\sigma_k(\vx_0)_1}{\sqrt {a(k+1)}}
\end{equation}
and
\begin{equation} \label{cla:hR2}
    \norms{\vh_{T_{01}}} \le \frac{2}{\sqrt{\theta_l}- \sqrt{\theta_u}/\sqrt a} \cdot \xkh{ \epsilon+  \frac{\sqrt{\theta_u} \sigma_k(\vx_0)_1}{\sqrt{a(k+1)}}}.
\end{equation}
Here, $C,\, c,\, \theta_l$ and $\theta_u$ are positive constants depending only on $\alpha$ and $\beta$.
Putting \eqref{cla:hR1} and \eqref{cla:hR2} into \eqref{cla:hR0}, we obtain that
\[
    \norms{\vh} \le \frac{2\xkh{1+1/\sqrt a}}{\sqrt{\theta_l}- \sqrt{\theta_u}/\sqrt a} \epsilon +  \xkh{\frac{2(1+1/\sqrt a)\sqrt{\theta_u}}{\sqrt{\theta_l}- \sqrt{\theta_u}/\sqrt a} +2} \frac{\sigma_k(\vx_0)_1}{\sqrt {a(k+1)}}.
\]

It remains to prove the claim \eqref{cla:hR1} and \eqref{cla:hR2}.   Since $\x$ is the solution to $\ell_1$ minimization program \eqref{eq:probset22}, we have
\begin{align*}
    \norm{\vx'_0}_1 \ge \norm{\x'}_1 & = \norm{\vx'_0 + \vh}_1
    = \norm{(\vx'_0 + \vh)_{T_0}}_1 + \norm{(\vx'_0 + \vh)_{T_0^c}}_1 \\
    & \ge \norm{\vx'_{0,T_0}}_1 - \norm{\vh_{T_0}}_1 + \norm{\vh_{T_0^c}}_1 - \norm{\vx'_{0,T_0^c}}_1.
\end{align*}
Therefore,
\begin{equation} \label{eq_MTrn_1}
    \|\vh_{T_0^c}\|_1 \le \norm{\vh_{T_0}}_1 + 2\norm{\vx'_{0,T_0^c}}_1.
\end{equation}
From the definition of $T_j$, we obtain that, for all $j \ge 2$,
\begin{equation*}
    \norm{\vh_{T_j}}_2 \le \sqrt{a(k+1)}\norm{\vh_{T_j}}_\infty = \frac{a(k+1)}{\sqrt {a(k+1)}}\norm{\vh_{T_j}}_\infty \le \frac{\norm{\vh_{T_{j-1}}}_1}{\sqrt {a(k+1)}}.
\end{equation*}
It then gives
\begin{equation} \label{eq_MTrn_2}
    \norm{\vh_{T_{01}^c}}_2
    \le \sum_{j \ge 2}\norm{\vh_{T_j}}_2 \le \frac{1}{\sqrt {a(k+1)}}\sum_{j \ge 2}\norm{\vh_{T_{j-1}}}_1
    = \frac{1}{\sqrt {a(k+1)}}\norm{\vh_{T_0^c}}_1.
\end{equation}
Putting \eqref{eq_MTrn_1} into \eqref{eq_MTrn_2}, we obtain the conclusion of claim \eqref{cla:hR1}, namely,
\begin{equation} \label{eq_MTrn_3}
    \begin{aligned}
        \norm{\vh_{T_{01}^c}}_2 & \le \frac{1}{\sqrt {a(k+1)}}\norm{\vh_{T_0^c}}_1
        \le \frac{\norm{\vh_{T_0}}_1 + 2\norm{\vx'_{0,T_0^c}}_1}{\sqrt {a(k+1)}} \\
        & \le\frac1{\sqrt a} \norm{\vh_{T_0}}_2 + \frac{2\sigma_{k+1}(\vx'_0)_1}{\sqrt {k}}
        \le \frac1{\sqrt a}\norm{\vh_{T_{01}}}_2 + \frac{2\sigma_k(\vx_0)_1}{\sqrt {a(k+1)}},
    \end{aligned}
\end{equation}
where the third inequality follows the Cauchy-Schwarz inequality
and the last inequality comes from the fact $\sigma_{k+1}(\vx'_0)_1 \le \sigma_k(\vx_0)_1$ by the definitions of $\x'$ and $\sigma_k(\cdot)_1$.

We next turn to prove the claim \eqref{cla:hR2}. Observe that
\begin{equation} \label{eq:Ahdiff}
    \norm{\bm A'_I \vh}_2 \ge \norm{\bm A'_I \vh_{T_{01}}}_2 - \norm{\bm A'_I \vh_{T_{01}^c}}_2.
\end{equation}
For the left hand side of \eqref{eq:Ahdiff}, by the definition of $I$, we have
\begin{equation} \label{eq:Ahdiff1}
\begin{aligned}
\norm{\bm A'_I \vh}_2 &= \norm{|\bm A'_I \x'| - |\bm A'_I \vx'_0|}_2 \\
& \le \norm{|\bm A' \x'| - |\bm A' \vx'_0|}_2 \\
& \le \norms{|\bm A' \x'|- \vy} +\norms{|\bm A' \vx'_0|- \vy} \\
& \le 2\epsilon.
\end{aligned} 
\end{equation}
For the first term of the right hand side of \eqref{eq:Ahdiff}, note that $|I|\ge m/2$.
According to Lemma \ref{lem_srip2},  we obtain that if $m\ge C (a+1)(k+1)\log(\e n/k)$, then with probability at least $1-4\exp(-c m)$ it holds
\begin{equation} \label{eq:Ahdiff2}
    \norm{\bm A'_I \vh_{T_{01}}}_2 \ge \sqrt{\theta_l} \norms{\vh_{T_{01}}}.
\end{equation}
To give an upper bound for the term $ \norm{\bm A'_I \vh_{T_{01}^c}}_2$, note that $\norm{\vh_{T_{01}^c}}_\infty \le \norm{\vh_{T_1}}_1/a(k+1)$.
Let $\theta:=\max\xkh{\norm{\vh_{T_1}}_1/a(k+1), \norm{\vh_{T_{01}^c}}_1/a(k+1) }$.
Then by the Lemma \ref{le:spadem}, we could decompose the vector $\vh_{T_{01}^c}$ into the following form:
\[
\vh_{T_{01}^c} =\sum_{j=1}^N \lambda_j \vu_j, \quad \text{with} \quad 0\le \lambda_j\le 1, \quad \sum_{j=1}^N \lambda_j=1,
\]
where $\vu_j$ are $a(k+1)$-sparse vectors satisfying
\[
\normone{\vu_j}=\normone{ \vh_{T_{01}^c}}, \quad \norm{\vu_j}_\infty \le \theta.
\]
Therefore, we have
\[
\norms{\vu_j} \le \sqrt{\theta \normone{ \vh_{T_{01}^c}}}.
\]
We notice from \eqref{eq_MTrn_1} that
\[
\normone{ \vh_{T_{01}^c}} \le \normone{ \vh_{T_{0}^c}} \le \norm{\vh_{T_0}}_1 + 2\sigma_k(\vx_0)_1.
\]
Thus, if $\theta = \norm{\vh_{T_1}}_1/a(k+1)$, then we have
\[
\norms{\vu_j} \! \le \! \sqrt{\frac{\norm{\vh_{T_1}}_1 \normone{\vh_{T_{01}^c}}}{a(k+1)}} \! \le \! \sqrt{\frac{\norm{\vh_{T_0^c}}_1 \normone{ \vh_{T_{01}^c}}}{a(k+1)}}
\! \le \! \frac{\norm{\vh_{T_0}}_1 + 2\sigma_k(\vx_0)_1}{\sqrt{a(k+1)}} \! \le \! \frac{\norm{\vh_{T_0}}_2}{\sqrt{a}}+\frac{ 2\sigma_k(\vx_0)_1}{\sqrt{a(k+1)}}.
\]
If $\theta=\norm{\vh_{T_{01}^c}}_1/a(k+1) $, then
\[
\norms{\vu_j} \le \frac{\norm{\vh_{T_{01}^c}}_1}{\sqrt{a(k+1)}} \le \frac{\norm{\vh_{T_0}}_2}{\sqrt{a}}+\frac{2\sigma_k(\vx_0)_1}{\sqrt{a(k+1)}}.
\]
Therefore, for the second term of the right hand side of \eqref{eq:Ahdiff}, according to Lemma \ref{lem_srip2}, we obtain that if $m\ge C (a+1)(k+1)\log(\e n/k)$, then with probability at least $1-4\exp(-c m)$ it holds
\begin{equation} \label{eq:Ahdiff3}
    \norm{\bm A'_I \vh_{T_{01}^c}}_2 =\norms{\sum_{j=1}^N \lambda_j \bm A'_I \vu_j } \le \sqrt{\theta_u} \sum_{j=1}^N \lambda_j \norms{\vu_j} \le \sqrt{\theta_u} \xkh{\frac{\norm{\vh_{T_0}}_2}{\sqrt{a}}+\frac{2\sigma_k(\vx_0)_1}{\sqrt{a(k+1)}}}.
\end{equation}
Putting \eqref{eq:Ahdiff1}, \eqref{eq:Ahdiff2} and \eqref{eq:Ahdiff3} into \eqref{eq:Ahdiff}, we immediately obtain that if $m\ge C a(k+1)\log(\e n/k)$, then with probability at least $1-4\exp(-c m)$ we have
\[
2\epsilon \ge \sqrt{\theta_l} \norms{\vh_{T_{01}}} - \sqrt{\theta_u}\xkh{\frac{\norm{\vh_{T_{01}}}_2}{\sqrt{a}}+\frac{2\sigma_k(\vx_0)_1}{\sqrt{a(k+1)}}},
\]
which gives
\[
\norms{\vh_{T_{01}}} \le \frac{2}{\sqrt{\theta_l}-\sqrt{\theta_u}/\sqrt a} \cdot \xkh{\epsilon + \frac{\sqrt{\theta_u} \sigma_k(\vx_0)_1}{\sqrt{a(k+1)}}}.
\]

{\bf Case 2: $|I| < m/2$}.  For this case, denote  $\vh^+ = \x' + \vx'_0$.  Replacing $\vh$ and the subset $I$ in Case 1 by $\vh^+$ and $I^c$ respectively, and applying the same argument,
we could obtain that when $m\ge C (a+1)(k+1)\log(\e n/k)$, with probability at least $1-4\exp(-c m)$, it holds
\begin{equation} \label{eq:h+}
\norm{\vh_+} \le \frac{2\xkh{1+1/\sqrt a}}{\sqrt{\theta_l}-\sqrt{\theta_u}/\sqrt a} \epsilon  + \xkh{\frac{2(1+1/\sqrt a)\sqrt{\theta_u}}{\sqrt{\theta_l}-\sqrt{\theta_u}/\sqrt a} +2} \frac{\sigma_k(\vx_0)_1}{\sqrt{a(k+1)}}.
\end{equation}
However, recall that $\x' = (\x^\T, 1)^\T$ and $\vx'_0 = (\bm x_0^\T, 1)^\T$.  It means $\norm{\vh_+}_2 \ge 2$, which contradicts to \eqref{eq:h+} by the assumption of $\epsilon$ and $\sigma_k(\vx_0)_1$, i.e., 
$ K_1\epsilon + K_2 \frac{\sigma_k(\vx_0)_1}{\sqrt {a(k+1)}} <2 $.
Therefore, Case 2 does not hold.

Combining the above two cases, we complete our proof.
\end{proof}

\section{Proof of Theorem \ref{lem_ripc20} and Theorem \ref{th:stubon0}}

\subsection{Proof of Theorem \ref{lem_ripc20}}

\begin{proof}
Without loss of generality, we assume that $\normf{\bm H'}=1$.
Observe that
\[
\frac{1}{m}\norm{\A'(\bm H')}_1=\frac{1}{m} \sum_{j=1}^m \abs{\va_j^* \bm H \va_j+ 2(b_j ( \va_j^*\vh))_{\Re}} := \frac{1}{m} \sum_{j=1}^m \xi_j.
\]
For any fixed $\bm H \in \mathbb H^{n\times n}$ and $\vh \in \C^n$, the terms $\xi_j, j=1,\ldots,m$ are independent sub-exponential random variables with the maximal sub-exponential norm
\[
K:=\max_{1\le j\le m}C_1(\normf{\bm H}+ \norm{b_j}_{\psi_2} \norm{\vh}) \le C_2
\]
for some universal constants $C_1,\, C_2>0$. Here, we use the fact $\max\xkh{\normf{\bm H},\norm{\vh}} \le \normf{\bm H'} = 1$.
For any $0<\epsilon \le 1$, the Bernstein's inequality gives
\begin{equation*}
\PP\Biggl(\Bigl| \frac{1}{m} \sum_{j=1}^m \xkh{\xi_j  - \E \xi_j} \Bigr|\ge \epsilon\Biggr) \le 2\exp\xkh{-c\epsilon^2 m},
\end{equation*}
where $c>0$ is a universal constant. According to Lemma \ref{lem_bounds}, we obtain that
\[
\frac 13 \E \sqrt{ \normf{\bm H}^2+ |b_j|^2 \norm{\vh}^2 } \le \E \xi_j \le 2 \E \sqrt{3 \normf{\bm H}^2+ |b_j|^2 \norm{\vh}^2}.
\]
This gives
\[
    \frac{1}{m} \sum_{j=1}^m \E \xi_j \le \frac{2}{m} \sum_{j=1}^m \E \xkh{\sqrt 3 \normf{\bm H}+ |b_j| \norm{\vh}}
    \le 2\sqrt 3 \normf{\bm H}+2c_2\norm{\vh} \le 2\theta^+,
\]
where $\theta^+:=\max(\sqrt 6, c_2)$. Here, we use the fact $\normf{\bm H'}^2=\normf{\bm H}^2+2\norm{\vh}^2=1$, $\E\norm{\vb}_1\le \sqrt m \E\norm{\vb} \le c_2m$,
and $\frac{a+b}{\sqrt 2} \le \sqrt{a^2+b^2} \le a+b$ for any positive number $a,b \in \R$.
Similarly, we could obtain
\[
\frac{1}{m} \sum_{j=1}^m \E \xi_j \ge \frac{1}{3\sqrt 2}\cdot \frac 1m \sum_{j=1}^m \E\xkh{\normf{\bm H}+|b_j|\norm{\vh}}
\ge \frac{1}{3\sqrt 2} \xkh{\normf{\bm H}+c_1 \norm{\vh}} \ge \frac{\theta^-} 6,
\]
where $\theta^-:=\min(1, c_1/\sqrt 2)$.
Collecting the above estimators, we obtain that,
with probability at least $1-2\exp(-c\epsilon^2 m)$, the following inequality
\begin{equation} \label{eq:fixHh}
\frac{\theta^-} 6 -\epsilon \le \frac{1}{m}\norm{\A'(\bm H')}_1\le 2\theta^+  +\epsilon
\end{equation}
holds for a fixed $\bm H' \in \mathbb H^{(n+1)\times (n+1)}$.
We next show that \eqref{eq:fixHh} holds for all $\bm H' \in \mathcal X$, where
\[
\mathcal X \!:=\! \dkh{\bm H'\!:=\!\begin{bmatrix} \bm H & \vh \\ \vh^* & 0 \end{bmatrix} \in \mathbb H^{(n+1)\times (n+1)}: \normf{\bm H'}=1,\, \rank(\bm H) \le 2,\, \norm{\bm H}_{0, 2} \le k,\, \norm{\vh}_0 \le k}.
\]
To this end, we adopt a basic version of a $\delta$-net argument. Assume that $\mathcal N_\delta$ is a $\delta$-net of $\mathcal X$,
i.e., for any $\bm H'=\begin{bmatrix} \bm H & \vh \\ \vh^* & 0 \end{bmatrix} \in \mathcal X$
there exists a $\bm H_0':=\begin{bmatrix} \bm H_0 & \vh_0 \\ \vh_0^* & 0 \end{bmatrix} \in \mathcal N_\delta$ such that $\normf{\bm H - \bm H_0}\le \delta$ and $\norm{\vh-\vh_0}\le \delta$.
Using the same idea of Lemma 2.1 in \cite{pr_2021}, we obtain that the covering number of $\mathcal X$ is
\[
|\mathcal N_\delta| \le \xkh{\frac{9\sqrt 2 \e n}{\delta k} }^{4k+2}\cdot \left(\begin{array}{l} n \\ k \end{array}\right) \xkh{1+\frac 2{\delta}}^{2k}  \le \exp\xkh{C_3 k\log (\e n/ \delta k)},
\]
where $C_3>0$ is a universal constant.
Note that $\vh-\vh_0$ has at most $2k$ nonzero entries. We obtain that if $m\gtrsim k\log(\e n/k)$, then with probability at least $1-3\exp(-c m)$, it holds
\begin{align*}
\abs{\frac{1}{m}\norm{\A'(\bm H')}_1-\frac{1}{m}\norm{\A'(\bm H_0')}_1}
& \le \frac{1}{m}\norm{\A'(\bm H'-\bm H'_0)}_1 \\
& \le \frac{1}{m}\norm{\mathcal A(\bm H-\bm H_0)}_1 +\frac 2 m \sum_{j=1}^m |b_j| |\va_j^* (\vh-\vh_0)| \\
& \le \frac{1}{m}\norm{\mathcal A(\bm H-\bm H_0)}_1 + 2\sqrt{\frac 1m \sum_{j=1}^m |b_j|^2} \! \sqrt{\frac 1m \sum_{j=1}^m  |\va_j^* (\vh-\vh_0)|^2}\\
& \le 2.45 \normf{\bm H-\bm H_0} + 3(c_2+1) \norm{\vh-\vh_0}\\
& \le 3\xkh{c_2+2} \delta,
\end{align*}
where the linear map $\mathcal A(\cdot)$ is defined as $\mathcal{A}(\bm H) := (\va_1^*\bm H \va_1, \dots, \va_m^*\bm H \va_m)$,
and the fourth inequality follows from the combination of Lemma \ref{lem_ripc}, the fact $ \frac 1m \sum_{j=1}^m \va_j \va_j^* \le 3/2$ with probability at least $1-\exp(-c m)$, and
\[
\frac 1m \sum_{j=1}^m |b_j|^2 \le \frac{\E \norm{\vb}^2}{m}+1\le c_2 +1
\]
with probability at least $1-2\exp(-c m)$. Choosing $\epsilon :=\frac{1}{48}$, $\delta:=\frac{\theta^-}{48(c_2+2)}$, and taking the union bound, we obtain that the following inequality
\[
\frac{\theta^-} {12}  \le \frac{1}{m}\norm{\A'(\bm H')}_1\le 3\theta^+ \quad \mbox{for all} \quad \bm H'\in \mathcal X
\]
holds with probability at least
\[
1-3\exp(-cm)-2 \exp\xkh{C_3 k \log (\e n/ \delta k)}\cdot\exp(-c \epsilon^2 m) \ge 1-5 \exp(-c' m),
\]
provided $m\ge C' k\log(\e n/k)$, where $C',\, c'>0$ are constants depending only on $c_1$ and $c_2$.
\end{proof}

\subsection{Proof of Theorem \ref{th:stubon0}}

\begin{proof}
The proof of this theorem is adapted from that of Theorem 1.3 in \cite{pr_2021}. Note that the $\ell_1$-minimization problem we consider is
\begin{equation} \label{eq:compra}
\argmin{\vx \in \C^n} \quad \norm{\vx}_1 \quad \mathrm{s.t.} \quad \norm{\mathcal A'(\vx')-\vy'}\le \epsilon \quad \mbox{with} \quad \vx'= \begin{pmatrix} \vx \\ 1\end{pmatrix}.
\end{equation}
Here, with some abuse of notation, we set
\[
\mathcal A'(\vx'):=\mathcal A'(\vx' \vx') = \xkh{|\va_1'^* \vx'|^2, \ldots, |\va_m'^* \vx'|^2} \quad \mbox{with} \quad \va_j=\begin{pmatrix} \va_j\\b_j\end{pmatrix},\quad  j=1,\ldots,m.
\]
Let $\x \in \C^n$ be a solution to \eqref{eq:compra}. Without loss of generality, we assume $\nj{\x',\vx_0'} \ge 0$ (Otherwise, we can choose $e^{\iu\theta} \vx_0'$ for an appropriate $\theta$),
where $\x'= \begin{pmatrix} \x \\ 1\end{pmatrix}$ and $\vx_0'= \begin{pmatrix} \vx_0 \\ 1\end{pmatrix}$. Set
\[
\hat{\bm X}':=\x'\x'^*=\begin{pmatrix} \x\x^* & \x\\ \x^*& 1 \end{pmatrix}
\]
and
\[
\bm H':=\x'\x'^*-\vx'_0{\vx'_0}^*=\begin{pmatrix} \x\x^* -\vx_0\vx_0^* & \x -\vx_0 \\ \x^* -\vx_0^*& 0 \end{pmatrix} :=\begin{pmatrix} \bm H & \vh \\ \vh^*& 0\end{pmatrix}.
\]
Therefore, it suffices to give an upper bound for $\normf{\bm H'}$. Denote $T_0:=\supp(\vx_0)$ and $T'_0:=T_0 \cup \dkh{n+1} $.
Let $T_1$ be the index set corresponding to the indices of the $ak$-largest elements of $\x_{T_0^c}$ in magnitude, and $T_2$ contain the indices of the next $ak$ largest elements, and so on.
Set $T_{01}:=T_0 \cup T_1$, $T'_{01}:=T'_0 \cup T_1$,   $\bh:=\vh_{T_{01}}$, $\bH=\bm H_{T_{01}, T_{01}}$, and $\bH':= \bm H'_{T'_{01}, T'_{01}}$.
Noting that
\begin{equation}\label{eq:HbH}
\normf{\bm H'}\le  \normf{\bH'} +\normf{\bm H'-\bH'},
\end{equation}
we next consider the terms  $\normf{\bH'}$ and $\normf{\bm H'-\bH'}$.
We claim that
\begin{equation} \label{cla:bH}
\normf{\bm H'-\bH'} \le  \xkh{ \frac1a +\frac{4}{\sqrt a}} \normf{\bH'}
\end{equation}
and
\begin{equation} \label{cal:Hprime}
\normf{\bH'} \le \frac 1{c - C\xkh{ \frac{4}{\sqrt a}+\frac 1a}} \cdot \frac{2\epsilon}{\sqrt m}.
\end{equation}
Combining (\ref{eq:HbH}), (\ref{cla:bH}) and (\ref{cal:Hprime}), we obtain that
\[
\normf{\bm H'}\le \frac {\frac{1}{a} +\frac{4}{\sqrt a}+1}{c - C\xkh{ \frac{4}{\sqrt a}+\frac 1a}} \cdot \frac{2\epsilon}{\sqrt m}.
\]
According to Lemma \ref{le:uuvv}, we immediately have
\[
\min_{\theta \in \R} \norms{\x'-\e^{\iu\theta} \vx'_0} \le  \frac{\sqrt 2\norm{\bm H'}}{\norm{\vx_0}+1} \le  \frac {\frac1a +\frac{4}{\sqrt a}+1}{c - C\xkh{ \frac{4}{\sqrt a}+\frac 1a}} \cdot \frac{2\sqrt 2 \epsilon}{\xkh{\norm{\vx_0}+1}\sqrt m}.
\]
By the definition of $\x'$ and $\vx'_0$, we arrive at the conclusion.

It remains to prove the claims \eqref{cla:bH} and \eqref{cal:Hprime}. Note that
\begin{equation} \label{eq:diffH1}
\normf{\bm H'-\bH'} \le \sum_{i\ge 2, j\ge 2} \normf{\bm H_{T_i, T_j}} + 2 \sum_{j\ge 2} \normf{\bm H'_{T'_0, T_j}}+2 \sum_{j\ge 2} \normf{\bm H'_{T_1, T_j}} .
\end{equation}
We first give an upper bound for the term $\sum_{i\ge 2, j\ge 2} \normf{\bm H'_{T_i, T_j}}$.
Noting that $\vx_0$ is a $k$-sparse vector and $\x \in \C^n$ is the solution to \eqref{eq:compra}, we obtain that
\[
\norm{\vx_0}_1 \ge \normone{\x}=\normone{\x_{T_0}}+\normone{\x_{T_0^c}},
\]
which implies $\normone{\x_{T_0^c}} \le \normone{\x_{T_0}-\vx_0}$.
Moreover, by the definition of $T_j$, we know that for all $j\ge 2$, it holds $ \norm{\x_{T_j} }_2 \le \frac{\norm{\x_{T_{j-1}}}_1}{\sqrt{ak}}$. It then implies
\begin{equation} \label{eq:xhaT0}
\sum_{j\ge 2} \norms{\x_{T_{j}}} \le \frac 1{\sqrt{ak}} \sum_{j\ge 2} \normone{\x_{T_{j-1}}}
\le  \frac 1{\sqrt{ak}}  \norm{\x_{T_0^c}}_1 \le \frac{1}{\sqrt a}  \norms{\x_{T_0}-\vx_0}.
\end{equation}
Therefore, the first term of \eqref{eq:diffH1} can be estimated as
\begin{equation} \label{eq:diffH2}
    \begin{aligned}
        \sum_{i\ge 2, j\ge 2} \normf{\bm H_{T_i, T_j}}= \sum_{i\ge 2, j\ge 2} \norms{\x_{T_i}} \norms{\x_{T_j}}  = \xkh{ \sum_{j\ge 2} \norms{\x_{T_{j}}}}^2 \le \frac 1{ak} \norm{\x_{T_{0}^c}}_1^2 \\
        = \frac 1{ak} \norm{\bm H_{T_0^c, T_0^c}}_1 \le \frac 1{ak} \norm{\bm H_{T_0, T_0}}_1 \le \frac1a \normf{\bH'},
        \end{aligned}
\end{equation}
where the second inequality follows from 
\[
\normone{\bm H-\bm H_{T_0,T_0}} =\normone{\x\x^*-(\x\x^*)_{T_0,T_0}} \le \normone{\vx_0{\vx_0}^*}-\normone{(\x\x^*)_{T_0,T_0}} \le \normone{\bm H_{T_0, T_0}}.
\]
Here, the first inequality comes from $\normone{\x} \le \normone{\vx_0}$.

For the second term and the third term of \eqref{eq:diffH1}, we obtain that
\begin{equation} \label{eq:diffH3}
\begin{aligned}
    \sum_{j\ge 2} \normf{\bm H'_{T'_0, T_j}}+\sum_{j\ge 2} \normf{\bm H'_{T_1, T_j}} &= \norm{\x'_{T'_0}}\sum_{j\ge 2} \norm{\x'_{T_j}} + \norm{\x'_{T_1}}\sum_{j\ge 2} \norm{\x'_{T_j}} \\
    &\le \frac{1}{\sqrt a} \norms{\x'_{T'_0}-\vx'_0} \xkh{\norms{\x'_{T'_0}}+\norms{\x'_{T_1}}} \\
    &\le \frac{\sqrt 2}{\sqrt a} \norms{\x'_{T'_{01}}-\vx'_0} \norms{\x'_{T'_{01}}} \\
    &\le \frac{2}{\sqrt a} \normf{\bH'},
\end{aligned}
\end{equation}
where the first inequality follows from \eqref{eq:xhaT0} due to $\x'_{T_j}=\x_{T_j}$ for all $j\ge 1$, and the last inequality comes from Lemma \ref{le:uuvv}.
Putting \eqref{eq:diffH2} and \eqref{eq:diffH3} into \eqref{eq:diffH1}, we obtain that
\[
\normf{\bm H'-\bH'} \le \xkh{ \frac1a +\frac{4}{\sqrt a}} \normf{\bH'}.
\]
This proves the claim \eqref{cla:bH}.

Finally, we turn to prove the claim \eqref{cal:Hprime}. Note that $\norm{\mathcal A'(\x')-\y}\le \epsilon$ and $\y:=\mathcal A'(\vx_0')+\epsilon$, which implies
\[
\norms{\mathcal A'(\bm H')} \le \norms{\mathcal A'(\x')-\y}+ \norms{\mathcal A'(\vx_0')-\y} \le 2\epsilon.
\]
Thus, we have
\begin{equation} \label{eq:tian1}
\frac{2\epsilon}{\sqrt m} \ge \frac 1{\sqrt m} \norms{\mathcal A'(\bm H')} \ge \frac1m \normone{\mathcal A'(\bm H')} \ge  \frac1m \normone{\mathcal A'(\bH')} - \frac1m \normone{\mathcal A'(\bm H'-\bH')}.
\end{equation}
Recall that $\bH':=\begin{pmatrix} \bH & \bh \\ \bh^*& 0\end{pmatrix}$ with $\rank(\bH)\le 2$, $\norm{\bH}_{0,2}\le (a+1)k$, and $\norm{\bh}_0 \le (a+1)k$. It then follows from the RIP of $\mathcal A'$  that
\begin{equation}  \label{eq:tian2}
\normone{\mathcal A'(\bH')}  \ge c \normf{\bH'}.
\end{equation}
To prove \eqref{cal:Hprime}, it suffices to give an upper bound for the term $\frac1m \normone{\mathcal A'(\bm H'-\bH')}$. Observe that
\begin{equation}  \label{eq:A0}
\bm H'-\bH' = (\bm H'_{T'_0, {T'^c_{01}}}+\bm H'_{{T'^c_{01}},T'_0}) + (\bm H'_{T_1, T'^c_{01}}+\bm H'_{ T'^c_{01},T_1})+\bm H'_{T'^c_{01}, T'^c_{01}}.
\end{equation}
Since
\[
\bm H'_{T'_0, {T'^c_{01}}}+\bm H'_{ {T'^c_{01}},T'_0} =\sum_{j\ge 2} (\bm H'_{T'_0, T_{j}}+\bm H'_{ T_{j},T'_0}) =\sum_{j\ge 2} \begin{pmatrix} \x_{T_0}\x_{T_j}^*+\x_{T_j}\x_{T_0}^* & \x_{T_j} \\  \x_{T_j}^* & 0\end{pmatrix},
\]
then the RIP of $\mathcal A'$ implies
\begin{equation} \label{eq:A1}
\begin{aligned}
\frac1m \normone{\mathcal A'(\bm H'_{T'_0, {T'^c_{01}}}+\bm H'_{ {T'^c_{01}},T'_0}  )} &\le C \sum_{j\ge 2} \xkh{ \normf{\x_{T_0}\x_{T_j}^*+\x_{T_j}\x_{T_0}^*} +2\norms{ \x_{T_j}}} \\
&\le 2\sqrt 2 C \norms{\x'_{T'_0}}  \sum_{j\ge 2} \norms{\x_{T_{j}}} \\
&\le \frac{2\sqrt 2}{\sqrt a}C \norms{\x'_{T'_0}}  \norms{\x'_{T'_{01}}-\vx'_0}.
\end{aligned}
\end{equation}
Similarly, we could obtain
\begin{equation} \label{eq:A2}
\frac1m \normone{\mathcal A'(\bm H'_{T_1, T'^c_{01}}+\bm H'_{ T'^c_{01},T_1})}  \le  \frac{2\sqrt 2}{\sqrt a}C \norms{\x'_{T_1}} \norms{\x'_{T'_{01}}-\vx'_0}.
\end{equation}
Finally, observe that  $\frac1m \normone{\mathcal A' (\bm H'_{T'^c_{01}, T'^c_{01}})} = \frac1m \normone{\mathcal A (\bm H_{T^c_{01}, T^c_{01}})}$. Using the same technique as \cite[Eq. (3.16)]{pr_2021}, we could obtain
\begin{equation}  \label{eq:A3}
\frac1m \normone{\mathcal A'(\bm H'_{T'^c_{01}, T'^c_{01}})}  \le \frac{C}{a} \normf{\bH'}.
\end{equation}
Putting \eqref{eq:A1}, \eqref{eq:A2} and \eqref{eq:A3} into \eqref{eq:A0}, we have
\begin{equation}  \label{eq:tian3}
\frac1m \normone{\mathcal A'(\bm H'-\bH')} \le \frac{4}{\sqrt a}C \norms{\x'_{T'_{01}}} \norms{\x'_{T'_{01}}-\vx'_0}+\frac{C}{a} \normf{\bH'} \le C\xkh{\frac{4}{\sqrt a}+\frac 1a}\normf{\bH'}.
\end{equation}
Combining \eqref{eq:tian1},  \eqref{eq:tian2} and  \eqref{eq:tian3}, we immediately obtain
\[
\xkh{c- C\xkh{ \frac{4}{\sqrt a}+\frac 1a}} \normf{\bH'} \le \frac{2\epsilon}{\sqrt m},
\]
which means
\[
\normf{\bH'} \le \frac 1{c- C\xkh{ \frac{4}{\sqrt a}+\frac 1a}} \cdot  \frac{2\epsilon}{\sqrt m}.
\]
This completes the proof of claim \eqref{cal:Hprime}.
\end{proof}

\appendix

\section{Supporting lemmas}

The following lemma gives a way for how to decompose a vector $\vv \in \R^n$ into the convex combination of several $k$-sparse vectors.

\begin{lem}[\cite{cai2013sparse,xu2013}] \label{le:spadem}
Suppose that $\vv \in \Rn$ satisfying $\norm{\vv}_\infty \le \theta$ and $\normone{\vv}\le k\theta$, where $\theta >0$ and $k\in \mathbb Z_+$. Then we have
\[
\vv=\sum_{j=1}^N \lambda_j \vu_j \quad \mathrm{with}\quad 0\le \lambda_j \le 1, \quad \sum_{j=1}^N \lambda_j=1,
\]
where $\vu_j \in \Rn$ is $k$-sparse vectors and $\normone{\vu_j} \le \normone{\vv}, \; \norm{\vu_j}_\infty \le \theta$.
\end{lem}

\begin{lem}[\cite{pr_2021}] \label{lem_ripc}
Let the linear map $\mathcal{A}(\cdot)$ be defined as
\begin{equation*}
\mathcal{A}(\bm H) := (\va_1^* \bm H \va_1, \dots, \va_m^*\bm H \va_m),
\end{equation*}
where $\va_j \sim 1/\sqrt{2}\cdot \mathcal N(0,I_n)+\mathrm i/\sqrt{2}\cdot \mathcal N(0,I_n), j=1,\ldots,m$ are i.i.d. complex Gaussian random vectors.
If $m \gtrsim k\log(n/k)$, then with probability at least $1 - 2\exp(-c_0m)$, $\mathcal{A}$ satisfies
\begin{equation*}
0.12\|\bm H\|_F \le \frac{1}{m}\|\mathcal{A}(\bm H)\|_1 \le 2.45\|\bm H\|_F
\end{equation*}
for all $H \in \mathbb{H}^{n \times n}$ with $\mathrm{rank}(\bm H) \le 2$ and $\|\bm H\|_{0, 2} \le k$. Here, $\|\bm H\|_{0, 2}$ denotes the number of non-zero rows in $\bm H$.
\end{lem}

\begin{lem}[\cite{pr_2021, npr_2020}] \label{le:uuvv}
For any vectors $\vu,\vv \in \C^n$ obeying $\nj{\vu,\vv} \ge 0$, we have
\[
\normf{\vu\vu^*-\vv\vv^*} \ge \frac1{\sqrt 2} \norms{\vu}\norms{\vu-\vv}.
\]
\end{lem}

\begin{lem} \label{lem_bjajx}
Suppose that $\va_j \sim \mathcal N(0,I_n), j=1, \dots, m$ are i.i.d.  Gaussian random vectors
and $\bm b\in \mathbb R^m$ is a nonzero vector.
For any fixed $\zeta \in (0, 1)$,
if $m \ge C\zeta^{-2}k(\log(\mathrm en/k)+\log(1/\zeta))$, then with probability at least $1 - 3\exp(-c_0 \zeta^2 m)$ it holds that
\begin{equation*}
\sum_{j=1}^m b_j(\va_j^\T \bm x) \le \zeta\sqrt{m}\norm{\bm x}_2\norm{\bm b}_2
\end{equation*}
for all $k$-sparse vectors $\bm x \in \mathbb R^n$. Here, $c_0 >0$ is a universal constant.
\end{lem}

\begin{proof}
Without loss of generality we assume $\norm{\bm x}_2 = 1$.
For any fixed $\bm x_0$, the terms $\va_j^\T \bm x_0$ are
independent, mean zero, sub-gaussian random variables with the maximal sub-gaussian norm being a positive universal constant.
The Hoeffding's inequality implies
\begin{equation*}
\PP\left(|b_j(\va_j^\T \bm x_0)| \ge t\right) \le 2\exp\bigl(-\frac{c_1^2t^2}{\norm{\bm b}_2^2}\bigr).
\end{equation*}
Here, $c_1 > 0$ is a universal constant.
Taking $t = \zeta\sqrt{m}\norm{\bm b}_2/2$, we obtain that
\begin{equation} \label{eq_bjajx}
\Bigl|\sum_{j=1}^m(\va_j^\T \bm x_0)\Bigr| \le \frac{\zeta}{2} \cdot \sqrt m \norm{\bm b}_2
\end{equation}
holds with probability at least $1 - 2\exp(-c_1\zeta^2m/4)$.

Next, we give a uniform bound to \eqref{eq_bjajx} for all $k$-sparse vectors $\bm x$.
Denote
\begin{equation*}
\mathcal S_{n,k} = \{\bm x\in \mathbb R^n: \norm{\bm x}_2=1, \norm{\bm x}_0 \le k\}.
\end{equation*}
We assume that $\mathcal N$ is a $\delta$-net of $\mathcal S_{n,k}$ such that for any $\bm x \in \mathcal S_{n,k}$, there exists a vector $\bm x_0 \in \mathcal N$ such that $\norm{\bm x - \bm x_0}_2 \le \delta$.
The covering number
$|\mathcal N| \le \left( \begin{array}{l} n\\k \end{array} \right) (1+\frac{2}{\delta})^k$. Note that $\norm{\vx-\vx_0} \le 2k$.
Therefore, when $m\gtrsim 2k$, with probability at least $1 - \exp(-c_2 m)$, it holds,
Thus we have
\begin{align*}
\left|\Bigl|\sum_{j=1}^m b_j(\va_j^\T \bm x)\Bigr| - \Bigl|\sum_{j=1}^m b_j(\va_j^\T \bm x_0)\Bigr|\right|
&\le \Bigl|\sum_{j=1}^m b_j \va_j^\T (\bm x-\bm x_0)\Bigr| \\
&\le \norm{\bm b}_2 \sqrt{\sum_{j=1}^m |\va_j^\T (\bm x - \bm x_0)|^2}\\
&\le \norm{\bm b}_2 \sqrt{\Bigl\|\sum_{j=1}^m \va_j\va_j^\T\Bigr\|_2} \cdot \norm{\bm x-\bm x_0}_2 \\
&\le 2\norm{\bm b}_2 \sqrt m \cdot \delta,
\end{align*}
where the second inequality follows from the Cauchy-Schwarz inequality and the last inequality comes from the fact
$\norm{\sum_{j=1}^m \va_j \va_j^\T}_2 \le 4m$ with probability at least $1 - \exp(-c_2 m)$, where $c_2>0$ is a universal constant.
Choosing $\delta = \zeta/4$ and taking the union bound over $\mathcal N$,
we obtain that
\begin{equation*}
\Big|\sum_{j=1}^m b_j(\va_j^\T \bm x_0)\Big| \le \zeta \cdot \sqrt m \norm{\bm b}_2
\end{equation*}
holds with probability at least
\begin{equation*}
1 - 2\exp(-c_1\zeta^2m/4)\cdot \left( \begin{array}{l} n\\k \end{array} \right) \cdot (1+\frac{2}{\delta})^{k} - \exp(-c_2 m) \ge 1-3\exp(-c \zeta^2 m)
\end{equation*}
provided $m \ge C \zeta^{-2} k (\log(\mathrm en/k)+\log(1/\zeta))$. Here, $C$ and $c$ are positive universal constants.
This completes the proof.
\end{proof}

\begin{lem} \label{lem_bounds}
Suppose that $\va \in \C^n$ is a complex Gaussian random vector and $b \in \C$ is a complex number. For any Hermitian matrix $\bm H \in \C^{n\times n}$ with $\rank(\bm H) \le 2$ and any vector $\vh \in \C^n$, we have
\[
\frac 13 \sqrt{ \normf{\bm H}^2+ b^2 \norm{\vh}^2 }  \le \E \abs{\va^* \bm H \va+ 2(  b( \va^*\vh))_{\Re}}  \le 2\sqrt{3 \normf{\bm H}^2+ b^2 \norm{\vh}^2}.
\]
\end{lem}

\begin{proof}
Since $\bm H \in \C^{n\times n}$ is a Hermitian matrix with  $\rank(\bm H) \le 2$, we can decompose $\bm H$ into
\begin{equation*}
\bm H= \lambda_1 \bm u_1 \bm u_1^* + \lambda_2 \bm u_2 \bm u_2^*,
\end{equation*}
where $\lambda_1, \lambda_2 \in \mathbb R$ are eigenvalues of $\bm H$ and $\bm u_1, \bm u_2 \in \mathbb C^n$ are the corresponding eigenvectors with $\norm{\bm u_1}_2 = \norm{\bm u_2}_2 = 1, \nj{\vu_1,\vu_2}=0$.
For the vector $\vh\in \C^n$, we can write it in the form of
\[
\vh=\sigma_1 \vu_1 +\sigma_2 \vu_2+\sigma_3 \vu_3,
\]
where $\sigma_1, \sigma_2, \sigma_3 \in \C$, and  $\vu_3 \in \C^n$ satisfying $\nj{\vu_3,\vu_1}=0, \nj{\vu_3,\vu_2}=0$ and $\norm{\vu_3}=1$.
For simplicity, without loss of generality, we assume that $b$ is a real number. Therefore,  we have
\[
\va^* H \va+ 2(b( \va^*\vh))_{\Re} = \lambda_1 |\va^* \vu_1|^2+\lambda_2 |\va^* \vu_2|^2 +2b \xkh{ \sigma_1 \va^*\vu_1+\sigma_2 \va^*\vu_2+\sigma_3 \va^*\vu_3}_{\Re}.
\]
Note that $\va \in \C^n$ is a complex Gaussian random vector and $\vu_1,\, \vu_2,\, \vu_3$ are orthogonal vectors. Thus, we have
\begin{equation*}
\E\abs{\va^* H \va+ 2(b(\va^*\vh))_{\Re}} = \E |\xi|,
\end{equation*}
with $\xi$ being a random variable given by
\[
\xi=\lambda_1 z_1^2+ \lambda_1 z_2^2+ \lambda_2 z_3^2+ \lambda_2 z_4^2+2b \xkh{ \sigma_{1,\Re} z_1 - \sigma_{1,\Im} z_2+ \sigma_{2,\Re} z_3 - \sigma_{2,\Im} z_4+\sigma_{3,\Re} z_5  - \sigma_{3,\Im} z_6}.
\]
Here, $z_1, z_2, z_3, z_4, z_5, z_6 \sim \mathcal N(0, 1/2)$ are independent. By Cauchy-Schwarz inequality, we have
\[
\E |\xi| \le \sqrt{\E \xi^2} \quad \mbox{and} \quad  \E \xi^2 = \E( \xi^{\frac{2}{3}}\xi^{\frac{4}{3}}) \le (\mathbb E\xi )^{\frac{2}{3}}(\mathbb E\xi_j^4)^{\frac{1}{3}}.
\]
It immediately gives
\begin{equation} \label{eq:absxi}
    \sqrt{\frac{(\mathbb E\xi^2)^3}{\mathbb E\xi^4}} \le \E|\xi| \le \sqrt{\E \xi^2}
\end{equation}
Let $z_1 = \rho_1\cos\theta$, $z_2 = \rho_1\sin\theta$, $z_3 = \rho_2\cos\phi$ and $z_4 = \rho_2\sin\phi$, $z_5 = \rho_3\cos\gamma$ and $z_6 = \rho_3\sin\gamma$.
Through some tedious calculations, we have
\begin{eqnarray*}
\E \xi^2 &= & \Bigl(\frac{1}{2\pi}\Bigr)^3 \int_0^{2\pi} \int_0^{2\pi}\int_0^\infty \int_0^{2\pi}\int_0^\infty\int_0^\infty
\rho_1\rho_2 \rho_3 \left( \lambda_1 \rho_1^2 +  \lambda_2 \rho_2^2 + 2b(\sigma_{1,\Re}  \rho_1\cos\theta - \sigma_{1,\Im}  \rho_1\sin\theta \right. \\
& & \left. + \sigma_{2,\Re} \rho_2\cos\phi - \sigma_{2,\Re}  \rho_2\sin\phi+ \sigma_{3,\Re}  \rho_3\cos\gamma - \sigma_{3,\Im}  \rho_3 \sin\gamma)\right)^2 \e^{-\frac{\rho_1^2+\rho_2^2+\rho_3^2}{2}}\mathrm d\rho_1\mathrm  d\rho_2\mathrm d\rho_3\mathrm d\theta\mathrm d\phi\mathrm d\gamma \\
&=&  8(\lambda_1^2 + \lambda_1\lambda_2 + \lambda_2^2) + 4b^2(\sigma_1^2+\sigma_2^2+\sigma_3^2)  \nonumber \\
&\le & 12 \normf{\bm H}^2+ 4b^2 \norm{\vh}^2,
\end{eqnarray*}
where the last inequality follows from the fact that $\lambda_1^2 + \lambda_2^2 = \norm{\bm H}_F^2$ and $\sigma_1^2+\sigma_2^2+\sigma_3^2=\norm{\vh}^2$.
Similarly, we could obtain
\begin{equation} \label{eq_bounds_lower_xi2}
\mathbb E\xi_j^2  \ge 4\norm{\bm H}_F^2 + 4b^2\norm{\vh}^2.
\end{equation}
and
\begin{equation} \label{eq:xi4}
\begin{aligned}
\E \xi^4 &=  48(8(\lambda_1^4 + \lambda_1^3\lambda_2 + \lambda_1^2\lambda_2^2 + \lambda_1\lambda_2^3 + \lambda_2^4)
+ b^4(\sigma_1^2+\sigma_2^2+\sigma_3^2 )^2 \\
& \quad + 4b^2(\lambda_1 + \lambda_2)^2 (\sigma_1^2+\sigma_2^2+\sigma_3^2)  + 8b^2 (\lambda_1^2\sigma_1^2+ \lambda_2^2\sigma_2^2) ) \\
& \le 48(12\norm{\bm H}_F^4 + b^4\norm{\vh}_2^4 + 16 b^2\norm{\bm H}_F^2\norm{\vh}_2^2) \\
& \le 576 \xkh{\normf{\bm H}^2+b^2\norm{\vh}^2}^2,
\end{aligned}
\end{equation}
where the first inequality follows from  the fact that
\[
\lambda_1^4 + \lambda_1^3\lambda_2 + \lambda_1^2\lambda_2^2 + \lambda_1\lambda_2^3 + \lambda_2^4
\le \lambda_1^4 + \lambda_1^2\lambda_2^2 + \lambda_2^4 + \frac 12 \xkh{\lambda_1^2+ \lambda_2^2}^2 \le \frac 23 \normf{\bm H}^4
\]
and
\[
\lambda_1^2\sigma_1^2+ \lambda_2^2\sigma_2^2
\le \xkh{\lambda_1^2+\lambda_2^2} \xkh{\sigma_1^2+\sigma_2^2+\sigma_3^2} \le \normf{\bm H}^2\norm{\vh}^2.
\]
Putting \eqref{eq_bounds_lower_xi2} and \eqref{eq:xi4}  into \eqref{eq:absxi}, we obtain
\[
\E|\xi| \ge \frac 13 \sqrt{ \normf{\bm H}^2+ b^2 \norm{\vh}^2 }.
\]

Therefore, we have
\[
\frac 13 \sqrt{ \normf{\bm H}^2+ b^2 \norm{\vh}^2 }  \le \E|\xi| \le 2\sqrt{3 \normf{\bm H}^2+ b^2 \norm{\vh}^2}.
\]
This completes the proof.
\end{proof}

\bibliographystyle{plain}

\begin{thebibliography}{10}
\bibitem{balan2006signal}
R. Balan, P. Casazza, and D. Edidin,
\newblock ``On signal reconstruction without phase,''
\newblock {\em  Appl. Comput. Harmon. Anal.},  vol. 20, no. 3, pp. 345--356, 2006.



\bibitem{barmherzig}
D. A. Barmherzig, J. Sun, and P. N. Li, T. J. Lane, E.J. Cand\`es,
\newblock ``Holographic phase retrieval and reference design,''
\newblock {\em  Inverse Probl.},  vol. 35, no. 9, 094001, 2019.


\bibitem{beinert2015}
R. Beinert  and G. Plonka,
\newblock ``Ambiguities in one-dimensional discrete phase retrieval from Fourier magnitudes,''
\newblock {\em J. Fourier Anal. Appl.}, vol. 21, no. 6, pp. 1169--1198, 2015.


\bibitem{beinert2018}
R. Beinert  and G. Plonka,
\newblock ``Enforcing uniqueness in one-dimensional phase retrieval by additional signal information in time domain,''
\newblock {\em Appl. Comput. Harmon. Anal.}, vol. 45, no.3,  pp. 505--525, 2018.


\bibitem{bendory}
T. Bendory, R. Beinert, and Y. C. Eldar,
\newblock ``Fourier phase retrieval: Uniqueness and algorithms,''
\newblock {\em Compressed Sensing and its Applications}, pp. 55--91, 2017.


%
\bibitem{saveing4d4}
 A. Bandeira, J. Cahill, D. Mixon, and A. Nelson,
\newblock ``Saving phase: Injectivity and stability for phase retrieval,"
\newblock {\em Applied and Computational Harmonic Analysis,} 37(1):106--125, 2014

\bibitem{2020a}
J. Cai, M. Huang, D. Li and Y. Wang,
\newblock `` Solving phase retrieval with random initial guess is nearly as good as by spectral
 initialization,"
\newblock {\em  Appl. Comput. Harmon. Anal.},  vol. 58,  pp. 60--84, 2022.


\bibitem{cai2013sparse}
T. T.  Cai and A. Zhang,
\newblock ``Sparse representation of a polytope and recovery of sparse signals and low-rank matrices,''
\newblock {\em IEEE Trans. Inf. Theory}, vol. 60, no. 1, pp. 122--132, 2013.


\bibitem{WF}
E. J.  Cand\`es, X. Li, and M. Soltanolkotabi,
\newblock ``Phase retrieval via Wirtinger flow: Theory and algorithms,''
\newblock {\em IEEE Trans. Inf. Theory}, vol. 61, no. 4, pp. 1985--2007, 2015.



\bibitem{TWF}
Y. Chen and  E. J. Cand\`es,
\newblock ``Solving random quadratic systems of equations is nearly as easy as
  solving linear systems,''
\newblock {\em Commun. Pure Appl. Math.}, vol. 70, no. 5, pp. 822--883, 2017.

\bibitem{edidin2019}
D. Edidin,
\newblock ``The geometry of ambiguity in one-dimensional phase retrieval,''
\newblock {\em  SIAM J. Appl. Algebr. Geom.},  vol. 3, no. 4, pp. 644--660, 2019.


\bibitem{eldar2014phase}
Y. C. Eldar and S. Mendelson,
\newblock `` Phase retrieval: Stability and recovery guarantees,"
\newblock {\em  Appl. Comput. Harmon. Anal.},  vol. 36, no.3,  pp. 473--494, 2014.


\bibitem{conca2015algebraic}
A. Conca, D. Edidin, M. Hering, and C. Vinzant,
\newblock ``An algebraic characterization of injectivity in phase retrieval,''
\newblock {\em  Appl. Comput. Harmon. Anal.},  vol. 38, no. 2, pp. 346--356, 2015.

\bibitem{fpr_1978}
J.~R. Fienup,
\newblock ``Reconstruction of an object from the modulus of its {Fourier}
  transform,''
\newblock {\em Optics Letters}, vol. 3, no. 1, pp. 27--29, 1978.

\bibitem{pra_1982}
J.~R. Fienup,
\newblock ``Phase retrieval algorithms: a comparison,''
\newblock {\em Applied Optics}, vol. 21, no. 15, pp. 2758--2769, 1982.

\bibitem{microscopy_1948}
D.~Gabor,
\newblock ``A new microscopic principle,''
\newblock {\em Nature}, vol. 161, no, 4098, pp. 777--778, 1948.
%
\bibitem{microscopy_1949}
D.~Gabor,
\newblock ``Microscopy by reconstructed wave-fronts,''
\newblock {\em Proceedings of the Royal Society of London. Series A.
  Mathematical and Physical Sciences}, vol. 197, no. 1051, pp. 454--487, 1949.


\bibitem{gao2016stable}
B. Gao, Y. Wang, and Z. Xu,
\newblock ``Stable signal recovery from phaseless measurements,''
\newblock {\em J. Fourier Anal. Appl.}, vol. 22, no. 4, pp. 787--808, 2016.




\bibitem{apr_2018}
B. Gao, Q. Sun, Y. Wang, and Z. Xu,
\newblock ``Phase retrieval from the magnitudes of affine linear measurements,''
\newblock {\em Advances in Applied Mathematics}, vol. 93, pp. 121--141, 2018.


\bibitem{guizar}
M. Guizar-Sicairos and J. R. Fienup,
\newblock ``Holography with extended reference by autocorrelation linear differential operation,''
\newblock {\em Opt. Express},  vol. 15, no. 26, pp. 17592--17612, 2007.



\bibitem{ppc_1993}
R.~W. Harrison,
\newblock ``Phase problem in crystallography,''
\newblock {\em Journal of the Optical Society of America A}, vol. 10, no. 5, 1993.

\bibitem{phase_1991}
H~A. Hauptman,
\newblock ``The phase problem of {X}-ray crystallography,''
\newblock {\em Reports on Progress in Physics}, vol. 54, no. 11, pp. 1427--1454, 1991.

%
%


\bibitem{huangK2016}
K. Huang, Y. C. Eldar, and N. D. Sidiropoulos,
\newblock ``Phase retrieval from 1D Fourier measurements: Convexity, uniqueness, and algorithms,''
\newblock {\em IEEE Trans. Signal Process.}, vol. 64, no. 23, pp.  6105--6117, 2016.


\bibitem{npr_2020}
M. Huang and Z. Xu,
\newblock ``Performance bound of the intensity-based model for noisy phase retrieval,''
\newblock arXiv preprint: {arXiv:2004.08764},  2020.

\bibitem{hx_2022}
M. Huang and Z. Xu,
\newblock ``Strong convexity of affine phase retrieval,''
\newblock arXiv preprint: {arXiv:2204.09412}, 2022.


\bibitem{latychevskaia}
T. Latychevskaia,
\newblock ``Iterative phase retrieval for digital holography: tutorial,''
\newblock {\em JOSA A}, vol. 36, no.12,  pp. 31--40, 2019.


\bibitem{liebling2003local}
M. Liebling, T. Blu, E. Cuche, P. Marquet, C. Depeursinge, and M. Unser,
\newblock ``Local amplitude and phase retrieval method for digital holography
  applied to microscopy,''
\newblock In {\em European Conference on Biomedical Optics}, vol. 5143, pp.210--214, 2003.

\bibitem{prco_1990}
R.~P. Millane.
\newblock ``Phase retrieval in crystallography and optics,''
\newblock {\em Journal of the Optical Society of Amerca A}, vol. 7, no .3, pp.394--411,
  March 1990.


\bibitem{AltMin}
P. Netrapalli, P. Jain,  and S. Sanghavi,
\newblock ``Phase retrieval using alternating minimization,''
\newblock {\em IEEE Trans. Signal Process.},  vol. 63, no. 18, pp.4814--4826,
  2015.


\bibitem{oss_2013}
J.~A. Rodriguez, R. Xu, C. Chen, Y. Zou, and J. Miao,
\newblock ``Oversampling smoothness: an effective algorithm for phase retrieval
  of noisy diffraction intensities,''
\newblock {\em Journal of Applied Crystallography}, vol. 46, no. 2, pp. 312--318, 2013.

\bibitem{fpr_1985}
J. L.~C. Sanz,
\newblock ``Mathematical considerations for the problem of Fourier transform
  phase retrieval from magnitude,''
\newblock {\em SIAM J. Appl. Math.},  vol. 45, no.4, pp. 651--664, 1985.


\bibitem{proi_2015}
Y. Shechtman, Y. C. Eldar, O. Cohen, H. N. Chapman, J. Miao, and M. Segev,
\newblock ``Phase retrieval with application to optical imaging: a contemporary
  overview,''
\newblock {\em IEEE Signal Process. Mag.}, vol. 32, no. 3, pp. 87--109, 2015.

\bibitem{turstregion}
J. Sun, Q. Qu, and J. Wright,
\newblock ``A geometric analysis of phase retrieval,''
\newblock {\em Found. Comput. Math.}, vol. 18, no. 5, pp.  1131--1198, 2018.


\bibitem{srip_2016}
V. Voroninski and Z. Xu,
\newblock ``A strong restricted isometry property, with an application to
  phaseless compressed sensing,''
\newblock {\em Appl. Comput. Harmon. Anal.},  vol. 40, no. 2, pp. 386--395,  2016.

\bibitem{pro_1963}
A. Walther,
\newblock ``The question of phase retrieval in optics,''
\newblock {\em J. Mod. Opt.}, vol. 10, no. 1, pp. 41--49, 1963.

\bibitem{pr_2021}
Y.~Xia and Z. Xu,
\newblock ``The recovery of complex sparse signals from few phaseless
  measurements,''
\newblock {\em Appl. Comput. Harmon. Anal.}, vol. 50, 2021.


\bibitem{TAF}
G. Wang, G. B. Giannakis, and Y.~C. Eldar,
\newblock ``Solving systems of random quadratic equations via truncated amplitude
  flow,''
\newblock {\em IEEE Trans. Inf. Theory},  vol. 64, no. 2, pp. 773--794, 2018.


\bibitem{wangxu}
Y. Wang and Z. Xu,
\newblock ``Generalized phase retrieval : measurement number, matrix recovery and beyond,''
\newblock {\em  Appl. Comput. Harmon. Anal.},  vol. 47, no. 2, pp. 423--446, 2019.

\bibitem{wangxu2}
Y. Wang and Z. Xu,
\newblock `` Phase Retrieval for Sparse Signals,''
\newblock {\em  Appl. Comput. Harmon. Anal.},  vol. 37, no. 3, pp. 531--544, 2014.


\bibitem{xu2013}
G. Xu and Z. Xu,
\newblock `` On the $\ell_1$-Norm Invariant Convex $k$-Sparse Decomposition of Signals,''
\newblock {\em  Journal of the Operations Research Society of China},  vol. 1, no. 4, pp. 537--541, 2013.


\bibitem{RWF}
H. Zhang, Y. Zhou, Y. Liang, and Y. Chi,
\newblock ``A nonconvex approach for phase retrieval: Reshaped wirtinger flow and incremental algorithms,''
\newblock {\em The Journal of Machine Learning Research}, vol. 18, no. 1, pp. 5164--5198, 2017.


\end{thebibliography}

\end{document}